\documentclass[11pt]{article}
\usepackage{amsmath,amssymb,amsthm}
\usepackage{paralist} 
\usepackage{times}
\usepackage{epsfig,epsf,psfrag,array,verbatim,mathtools}
\usepackage{url,comment}
\usepackage[usenames,dvipsnames,svgnames,table]{xcolor}

\usepackage[ruled,linesnumbered,vlined]{algorithm2e}

\usepackage{pdfpages,framed}
\usepackage{tikz}
\usetikzlibrary{arrows,patterns}
\usetikzlibrary{matrix,calc,decorations.pathreplacing}
\usepackage{color}

\theoremstyle{plain}
\newtheorem{theorem}{Theorem}
\newtheorem{lemma}{Lemma}

\newtheorem{problem}{Problem}
\newtheorem{myquestion}{Question}

\newcommand{\vol}{\textsc{vol}}
\newcommand{\parity}{\mathrm{parity}}
\newcommand{\textsum}{\mathrm{sum}}

\newcommand{\avd}{\textsc{Avoid}}

\newcommand{\set}[1]{\left\{ #1 \right\}}
\newcommand{\abs}[1]{\left| #1 \right|}
\newcommand{\paren}[1]{\left( #1 \right)}
\newcommand{\ceil}[1]{\left\lceil #1 \right\rceil}
\newcommand{\floor}[1]{\left\lfloor #1 \right\rfloor}

\newcommand{\s}{{\cal S}}

\def\R{\mathcal{R}}
\def\l1{\mathbf{L_1}}

\newcommand{\ia}{\alpha}
\newcommand{\ib}{\beta}

\long\def\comment #1\commentend{}
\def\inline#1:{\par\smallskip \noindent{\bf #1:}\hskip 5pt}

\title{Efficiently Realizing Interval Sequences}

\author{
Amotz Bar-Noy%
\thanks{City University of New York (CUNY), USA. E-mail: amotz@sci.brooklyn.cuny.edu},
Keerti Choudhary%
\thanks{Tel Aviv University, Israel. E-mail: keerti.choudhary@cs.tau.ac.il},
David Peleg%
\thanks{Weizmann Institute of Science, Israel. E-mail: david.peleg@weizmann.ac.il},
Dror Rawitz%
\thanks{Bar Ilan University, Israel. E-mail: dror.rawitz@biu.ac.il}
}
\date{}
\begin{document}
\maketitle
\begin{abstract}

We consider the problem of realizable interval-sequences. An interval sequence comprises of $n$ integer intervals $[a_i,b_i]$ such that $0\le a_i\leq b_i \le n-1$, and is said to be {\em graphic/realizable} if there exists  a graph with degree sequence, say, $D=(d_1,\ldots,d_n)$ satisfying the condition $a_i\leq d_i\leq b_i$, for each $i\in[1,n]$. There is a characterisation (also implying an $O(n)$ verifying algorithm) known for realizability of interval-sequences, which is a generalization of the Erd\"os-Gallai characterisation for graphic sequences. However, given any realizable interval-sequence, there is no known algorithm for computing a corresponding graphic certificate in $o(n^2)$ time.

In this paper, we provide an $O(n \log n)$ time algorithm for computing a graphic sequence for any realizable interval sequence. In addition, when the interval sequence is non-realizable, we show how to find a graphic sequence having minimum deviation with respect to the given interval sequence, in the same time. Finally, we consider variants of the problem such as computing the most regular graphic sequence, and computing a minimum extension of a length $p$ non-graphic sequence to a graphic one.

\end{abstract}

\section{Introduction}

The {\em Graph Realization problem} for a property $P$ 
deals with the following existential question:
Does there exist a graph that satisfies the property $P$?
Its fundamental importance is apparent, ranging from better theoretical
understanding, to network design questions
(such as constructing networks with certain desirable connectivity properties).
Some very basic, yet challenging, properties that have been considered in past
are degree sequences \cite{EG60,hakimi62,havel55},
eccentricites \cite{Behzad1976,Lesniak1975}, connectivity and flow~\cite{GH61,FC70,F92,F94}.

One of the earliest classical problems studied in this domain is that of
{\em graphic sequences}. A sequence of $n$ positive integers,
$D=(d_1,\ldots,d_n)$, is said to be {\em graphic} if there exists
an $n$ vertex graph $G$ such that $D$ is identical to the sequence of 
vertex degrees of $G$. 
The problem of realizing graphic sequences and counting the number of non-isomorphic 
realizations of a given graphic-sequence, is particularly of interest due to many practical
applications, see~\cite{Tatsuya2013} and reference therein.
In 1960, Erd\"os and Gallai~\cite{EG60} gave a 
characterization (also implying an $O(n)$ verifying algorithm) for graphic 
sequences. Havel and Hakimi~\cite{hakimi62,havel55} gave a recursive algorithm 
that given a sequence $D$ of integers computes a realizing graph,
or proves that the sequence is non-graphic, in optimal time $O(\sum_i d_i)$.
Recently, Tripathi et al.~\cite{TripathiVW:10} provided a constructive proof
of Erd\"os and Gallai's~\cite{EG60} characterization.

We consider a generalization of the graphic sequence
problem where instead of specifying precise degrees, we are given a
\emph{range} (or interval) of possible degree values for each vertex.
Formally, an \emph{interval-sequence} is a sequence of $n$ intervals
$\s =([a_1,b_1],\ldots,[a_n,b_n])$, also represented as $\s=(A,B)$, where $A=(a_1,\ldots,a_n)$ and
$B=(b_1,\ldots,b_n)$, and $0\le a_i \le b_i \le n-1$ for every $i$.
It is said to be \emph{realizable} if there exists a sequence
$D=(d_1,\ldots,d_n)$ that is graphic and satisfies the condition
$a_i\leq d_i\leq b_i$, for $1\le i\le n$. Two questions that are
natural to ask here are:

\begin{myquestion}[Verification]
Find an efficient algorithm for verifying the realizability of any
given interval-sequence $\s $?
\end{myquestion}

\begin{myquestion}[Graphic Certificate]
Given a realizable interval-sequence $\s $, compute a certificate
(that is, a graphic sequence $D$) realizing it.
\end{myquestion}

Cai et al.~\cite{CDZ:00} extended Erd\"os and Gallai's work by providing
an easy to verify characterization for realizable interval-sequences,
thereby resolving Question 1. Their result crucially uses the $(g,f)$-Factor
Theorem of Lov{\'{a}}sz~\cite{lovasz:70}.
Garg et al.~\cite{GargGT:11} provided a constructive proof of the
characterisation of Cai et al.~\cite{CDZ:00} for realizable interval sequences.
In~\cite{HellK:09}, Hell and Kirkpatrick provided an algorithm based on Havel
and Hakimi's work for computing a graph that realizes an interval sequence
(if exists). For non-realizable interval sequences $\s$, their algorithm computes
a graph whose deviation $\delta(D,\s)$ (see Section~\ref{section:prelim} for definition) 
with respect to L1-norm is minimum.
The time complexity of their algorithm is $O(\sum_{i=1}^n b_i)$
(which can be as high as $\Theta(n^2)$).

\subparagraph*{Our Contributions.}
In this paper we introduce a new approach for representing and
analyzing the interval sequence realization problem. 
Our algorithms are based on a novel divide and conquer methodology, 
wherein we show that partitioning a realizable interval sequence along any levelled sequence
(a new class of sequences introduced herein) guarantees that at least one
of the new child interval sequences is also realizable.
This enables us to present an $O(n \log n)$ time algorithm
for computing a graphic certificate (if exists) for any given interval sequence.
While the problem was well studied, to the best of our knowledge
there was no known $o(n^2)$ time algorithm for computing graphic certificate.
In addition, given an interval sequence $\s$, our algorithm can obtain in the same time a degree-certificate corresponding to graphs with minimum (resp. maximum) possible edges.
Specifically, we obtain the following result.

\begin{theorem}\label{theorem-our-result-1}
There exists an algorithm that for any integer $n\geq 1$ and any
length $n$ interval sequence $\s$, computes a graphic sequence $D$
realizing $\s$, if exists, in $O(n\log n)$ time.

Moreover, our algorithm can also output in the same time graphic certificates corresponding to a sparsest as well as densest possible graph 
(i.e. graphs with minimum and maximum possible edges), realizing~$\s$.
\end{theorem}

We also investigate the problem of efficiently computing graphic sequences having the least possible $L_1$-deviation in the scenario when the input interval sequence is non-realizable. We must point out here that till now there was also no sub-quadratic time algorithm known for computing even the deviation $\delta(D,\s)$. Our result for deviation minimizing certificate can be formalized as follows.

\begin{theorem}\label{theorem-our-result-1b}
There exists an algorithm that for any integer $n\geq 1$ and any
length $n$ non-realizable interval sequence $\s$, outputs in $O(n\log n)$ time a graphic
sequence $D$ minimizing the deviation $\delta(D,\s)$.
\end{theorem}

Our new approach enables us to tackle also an optimization version of
the problem in which it is required to compute the ``most regular''
sequence realizing the given interval sequence $\s$, using the natural
measure of the minimum sum of pairwise degree differences, $\sum_{i,j}
|d_i-d_j|$, as our regularity measure.  To the best of our knowledge,
this problem was not studied before and is not dealt with directly by
the existing approaches to the interval sequence problem. 
Specifically, we obtain the following.

\begin{theorem}\label{theorem-our-result-2}
There exists an algorithm that for any integer $n\geq 1$ and any length $n$
realizable interval sequence $\s$, computes the most regular graphic sequence
realizing interval sequence $\s$
(i.e., the one minimizing the sum of pairwise degree difference),
in time $O(n^2)$.
\end{theorem}

The tools developed in this paper allows us to study other interesting applications, such as computing a minimum
extension of non-graphic sequences to graphic ones (see Section~\ref{section:application}).

\subparagraph*{Related work.}
Kleitman and Wang~\cite{KleitmanW73}, and Fulkerson-Chen-Anstee~\cite{Anstee1982,Chen1966,Fulkerson1960}
solved the problem of degree realization for directed graphs,
wherein, for each vertex both the in-degree and out-degree is specified.
In \cite{HartungN15}, Nichterlein and Hartung proved the NP-completeness 
of the problem when the additional constraint of acyclicity is imposed.
Over the years, various extensions of the degree realization problems
were studied as well, cf. \cite{AT94,WK73}. 
The {\em Subgraph Realization problem} considers the restriction that the realizing graph 
must be a subgraph ({\em factor}) of some fixed input graph.
For an interesting line of work on graph factors, refer to~\cite{Tutte1981,ANSTEE:1985,Hell:1990,guo:hal-01179211}.
The subgraph realization problems are generally harder. For instance,
it is very easy to compute an $n$-vertex connected graph whose degree sequence consists of all values $2$,
however, the same problem for subgraph-realization is NP-hard (since it reduces to Hamiltonian-cycle problem).

Lesniak~\cite{Lesniak1975} provided a characterization for the sequence of eccentricities of an
$n$-vertex graph. Behzad et al.~\cite{Behzad1976}
studied the problem of characterizing the set  comprising of vertex-eccentricity values of general graphs
(the sequence problem remains open). 
Fujishige et al.~\cite{FujishigeP01} considered the problem of realizing graphs and hypergraphs with given cut specifications.
Realization problems related to various criteria of relative satisfaction are considered in~\cite{BCPR19happiness}. 
Several other realization problems are surveyed in~\cite{BCPR18survey}.

\subparagraph*{Organization of the Paper.}
In Section~\ref{section:prelim}, we present the notation and definitions.
In Section~\ref{section:tools}, we discuss the main ideas and tools
that help us to construct graph certificates for interval sequence problem.
Section~\ref{section:algo1} presents our $O(n\log n)$ time algorithm 
for computing graphic certificate with minimum deviation.
Section~\ref{section:algo2} provides a quadratic-time algorithm
for computing the most regular certificate.
We discuss the applications in Section~\ref{section:application}.
The appendix includes linear time verification algorithms implied by the 
work of Erd\"os and Gallai~\cite{EG60}, and Cai et al.~\cite{CDZ:00}.

\section{Preliminaries}
\label{section:prelim}

A \emph{sequence} is defined to be an $n$-element vector whose entries
are non-negative integers. For any sequence $D = (d_1,\ldots,d_n)$,
define $\min(D)=\min_{i=1}^n \{d_i\}$, $\max(D) = \max_{i=1}^n
\{d_i\}$, $\textsum(D) = \sum_{i=1}^n d_i$, and
$\parity(D)~=~\textsum(D)\mod 2$.  Given any two sequences $X =
(x_1,\ldots,x_n)$ and $Y=(y_1,\ldots,y_n)$, we say that $X \leq Y$ if
$x_i \leq y_i$ for $1\leq i\leq n$.  Any two sequences $X$ and~$Y$ are
said to be \emph{similar} if they are identical up to permutation of
the elements (i.e., their sorted versions are identical).  A sequence
$D$ is said to \emph{lie} in an interval-sequence $(A,B)$, denoted by
$D \in(A,B)$, if $A \leq D \leq B$.  We define
$\min(X,Y) = (\min\{x_1,y_1\}, \ldots, \min\{x_n,y_n\})$, and
$\max(X,Y) = (\max\{x_1,y_1\}, \ldots, \max\{x_n,y_n\})$.
The \emph{$\l1$}-distance of the pair $(X,Y)$ is defined as
$\l1(X,Y) ~=~ \sum_{i=1}^n|y_i-x_i|$.

Denote by $\top$ and $\bot$ the $n$-length sequences all whose
entries are respectively $n-1$ and $0$.  Given a sequence $D =
(d_1,\ldots,d_n)$ and an integer $k\in[1,n]$, define the vectors
$X(D)$ and $Y(D)$ by setting~for~$1 \le k \le n$:
$$~~~~X_k(D)  \triangleq \sum_{i=1}^k d_i,~~~ \text{and} ~~~
Y_k(D)  \triangleq k(k-1) + \sum_{i=k+1}^n \min(d_i,k)~.$$
%
%
For any sequence $D=(d_1,\ldots,d_n)$, the \emph{spread} of $D$ is
defined as $\phi(D)~=~\sum_{1\leq r<s\leq n} |d_r-d_s|,$ and it always
lies in the range $[0,n^3]$.  A sequence $D$ is said to be more
\emph{regular} than another sequence $D'$ if $\phi(D) < \phi(D')$.
For any two integers $x\leq y$,  $[x,y] =\{x,x+1,\ldots,y\}$.
For any $I \subseteq [1,n]$, define $D[I]$ to be the subsequence of
$D$ consisting of elements $d_i$, for $i\in I$; and define $E_{I}$ to
be the \emph{characteristic vector} of $I$, namely, the sequence
$(e_1,e_2,\ldots,e_n)$ such that $e_i=1$ if $i\in I$, and $e_i=0$
otherwise.
For any sequence $D=(d_1,\ldots,d_n)$ and an interval-sequence 
$\s =([a_1,b_1],\ldots,[a_n,b_n])$, the upper and lower deviation of $D$, 
is respectively defined as
$$\delta_U(D,\s ) =\sum_{i=1}^n \max\{0,(d_i-b_i)\}, ~~ \text{ and } ~~
\delta_L(D,\s ) = \sum_{i=1}^n \max\{0,(a_i-d_i)\}~.$$
%
%
The \emph{deviation} of $D$ is defined as $\delta(D,\s ) =
\delta_U(D,\s ) + \delta_L(D,\s )$.
For any vertex $x$ in an undirected simple graph $H$, define
$\deg_H(x)$ to be the degree of $x$ in $H$, and define $N_H(x) =
\{y~|~(x,y) \in E(H)\}$ to be the neighbourhood of $x$ in $H$.


We next state the Erd\"os and Gallai~\cite{EG60} characterisation for
realizable (graphic) sequences, and Cai et al.~\cite{CDZ:00} characterisation for
realizable interval sequences. An $O(n)$-time implementation of the both theorems is provided
in the Appendix~\ref{section:EG_CDZ_implementation}.

\begin{theorem}[Erd\"os and Gallai~\cite{EG60}]
\label{theorem:EG}
A non-increasing sequence $D=(d_1,\ldots,d_n)$ is graphic if and only if
\begin{inparaenum}[(i)]
\item $X_n(D)$ is even, and
\item $X(D)\leq Y(D)$.
\end{inparaenum}
\end{theorem}

\begin{theorem}[Cai et al.~\cite{CDZ:00}]
\label{theorem:CDZ}
Let $\s =([a_1,b_1],\ldots,[a_n,b_n])=(A,B)$ be an interval-sequence
such that $A$ is non-increasing and for any index $1 \leq i < n$,
$b_{i+1} \leq b_{i}$ whenever $a_i=a_{i+1}$.  For each $k \in [1,n]$,
define $W_k(\s) = \{i \in [k+1,n] ~|~ b_i \geq k+1\}$.  Then $\s$ is
realizable if and only if $X(A) \leq Y(B)-\varepsilon(\s)$, where,
$\varepsilon(\s)$ is defined by setting
$$
\varepsilon_k(\s ) =
\begin{cases}
1 & \text{if } a_i=b_i \text{ for } i\in W_k(\s ) \text{ and } \sum_{i\in W_k(\s )} (b_i+k|W_k(\s )|) \text{ is odd},\\
0 & \text{otherwise.}
\end{cases}
$$
\end{theorem}

\section{Main Tools}
\label{section:tools}

In this section, we develop some crucial tools that help us in
efficient computation of certificate for a realizable
interval-sequence. These tools will help us to search a graphic
sequence in $O(n\log n)$ time using a clever divide and conquor
methodology.  Also they aid in searching for the maximally regular
sequence in just quadratic time.

\subparagraph{Levelling operation.}

Given a sequence $D=(d_1,\ldots,d_n)$ and a pair of indices $\ia \neq
\ib$ satisfying $d_\ia > d_\ib$, we define $\pi(D,\ia,\ib) = D^* =
(d^*_1,\ldots,d^*_n)$ to be a sequence obtained from $D$ by
decrementing $d_\ia$ by $1$ and incrementing $d_\ib$ by 1 (i.e.,
$d^*_\ia = d_\ia-1$, $d^*_\ib = d_\ib+1$, and $d^*_k=d_k$ for $k\neq
\ia,\ib$).  This operation is called the \emph{levelling operation} on
$D$ for the indices $\ia$ and $\ib$.  The operation essentially
``levels'' (or ``flattens'') the sequence $D$, making it more uniform.

We now discuss some properties of levelling operations.


\begin{lemma}
Any levelling operation on a sequence $D$ that results in a
non-similar sequence, reduces its spread $\phi(D)$ by a value at least two.
\label{lemma:spread}
\end{lemma}

\begin{proof}
Let $D=(d_1,d_2,\ldots,d_n)$ and $Z=(z_1,\ldots,z_n)=\pi(D,\ia,\ib)$,
be a sequence obtained from $D$ by performing a levelling operation on
a pair of indices $\ia,\ib$ such that $d_\ia>d_\ib$.  If
$d_\ia=d_\ib+1$, then it is easy to verify that $D$ and $Z$ are
similar.  If $d_\ia\geq d_\ib+2$, then
\begin{eqnarray*}
\phi(Z)
& =    & |z_\ia-z_\ib| + \sum_{s\neq \ia,\ib} (|z_\ia-z_s|+|z_\ib-z_s|)
	 + \sum_{\substack{1\leq r<s\leq n,\\r,s\notin\{\ia,\ib\}}} |z_r-z_s|\\
& =    & |d_\ia-d_\ib|-2
         +\sum_{\substack{s \neq \ia, \ib~\mbox{\scriptsize s.t.}\\
                d_s\notin(d_\ib,d_\ia)}} (|d_\ia-d_s|+|d_\ib-d_s|) \\
&      & +\sum_{\substack{s \neq \ia, \\\ib~\mbox{\scriptsize s.t.}
                d_s\in(d_\ib,d_\ia)}} (|d_\ia-d_s|+|d_\ib-d_s|-2)
	  +\sum_{\substack{1\leq r<s\leq n,\\r,s\notin\{\ia,\ib\}}}|d_r-d_s|\\
& \leq &\Big(\sum_{1\leq r<s\leq n}|d_r-d_s|\Big)-2~=~ \phi(D)-2
~.
\end{eqnarray*}
Thus, the claim follows.
\end{proof}


\begin{lemma}[Corollary 3.1.4, \cite{mahadev1995threshold}]
The levelling operations preserves graphicity, that is, if we perform
a levelling operation on a graphic sequence, then the resulting
sequence is also graphic.
\label{lemma:redistribute}
\end{lemma}
\begin{proof}
Let $D=(d_1,\ldots,d_n)$ be a graphic sequence, and
$\pi(D,\ia,\ib)=D^*=(d^*_1,\ldots,d^*_n)$ for some indices $\ia,\ib$
satisfying $d_\ia>d_\ib$. If $d_\ia=1+d_\ib$, then $D^*$ is similar to
$D$, and thus also graphic.  So for the rest the proof let us focus on
the case $d_\ia\geq 2+d_\ib$.  Let $G=(V,E)$ be a graph realising the
sequence $D$, and let $x_\ia$ and $x_\ib$ be two vertices in $G$
having degrees respectively $d_\ia$ and $d_\ib$. Since
$|N_G(x_\ia)|\geq 2+|N_G(x_\ib)|$, there must exists at least one
neighbour, say $w$, of vertex $x_\ia$ that does not lie in set
$\{x_\ib\}\cup N_G(x_\ib)$. Let $G^*=(V,E^*)$ be a graph obtained from
$G$ by deleting the edge $(w,x_\ia)$, and adding a new edge
$(w,x_\ib)$. Observe that the degree of all vertices other than
$x_\ia$ and $x_\ib$ are identical in graphs $G$ and $G^*$, also
$\deg_{G^*}(x_\ia) = \deg_{G}(x_\ia)-1$, and $\deg_{G^*}(x_\ib) =
\deg_{G}(x_\ib)+1$.  Therefore $G^*$ is a graph realising the profile
$D^*$, and thus the claim follows.
\end{proof}

\subparagraph*{Levelled sequences.}

A sequence $D$ is said to be \emph{levelled with respect to the
  integer-sequence $\s =(A,B)$} if
\begin{inparaenum}[(i)]
\item $A\leq D\leq B$, and
\item the spread of $D$ cannot be decreased by a levelling operation,
  i.e., for any two indices $\ia\neq \ib$ satisfying $d_\ia>d_\ib$ and
  $A\leq \pi(D,\ia,\ib)\leq B$, we have
  $\phi(\pi(D,\ia,\ib))=\phi(D)$.
\end{inparaenum}
See Figure~\ref{Figure:levelled}.

The \emph{volume} of a sequence $D$ lying between $A$ and $B$ with
respect to $\s = (A,B)$ is defined as
$$
\vol(D,\s) \triangleq \l1 (D,A)
~,
$$
and is invariant of levelling operations applied to $D$.
In other words, applying a levelling operation to a sequence $D$
may reduce its spread but preserves its volume.
Note that the volume lies in the range $[0,~\l1 (A,B)]$.


\begin{lemma}
\label{lemma:making_levelled}
For any $\s =(A,B)$, a sequence $D$ satisfying $A\leq D\leq B$
can be transformed into a levelled sequence $D^*$ having the same
volume $\vol(D,\s)$ by a repeated application of
(at most $O(n^3)$) levelling operations.%
\footnote{We remark that the algorithms presented later on generate a
  desired levelled sequence using more efficient methods than the one
  implicit in the proof, and are therefore faster.  }
\end{lemma}

\begin{proof}
By Lemma~\ref{lemma:spread}, every levelling operation that results in
a new (non-similar) sequence decreases the spread by at least two.
Since the spread of any sequence $D$ is always non-negative and finite
(specifically, $O(n^3)$), it is possible to perform ($O(n^3)$)
levelling operations on $D$ so that the resultant sequence $D^*$ is
levelled.  Since the levelling operation preserves the volume,
$\vol(D^*,\s)$ must be same as $\vol(D,\s)$.
\end{proof}


Any graphic sequence $D$ realizing the interval sequence $\s =(A,B)$
by Lemma~\ref{lemma:making_levelled} can be altered by $O(n^3)$
levelling operations to obtain a levelled sequence lying between
$A$ and $B$.
The resultant sequence by Lemma~\ref{lemma:redistribute} remains
graphic, thus the following theorem is immediate.

\begin{theorem}
\label{theorem:levelled_graphic_sequence}
For any realizable interval sequence $\s =(A,B)$ there exists a
graphic sequence realizing $\s$ which is a levelled sequence.
\end{theorem}

\subparagraph{Characterizing and Computing Levelled sequences.}

Given any interval sequence $\s = (A,B)$
and a real number $\ell\in[\min(A),\max(B)]$, let%
\footnote{One can think of $\s $ as representing a collection of $n$ connected vessels,
each in the shape of a unit column closed at both ends, then $F(\ell,\s)$
is the amount of fluid that will fill this connected vessel system to level
$\ell$.}
$$
F(\ell,\s)
~\triangleq~
\sum_{i \in[1,n]} (\min\{\ell,b_i\} - \min\{\ell,a_i\})
~.
$$
Observe that $F(\cdot,\s )$ is a non-decreasing function in the range
$(\min(A),\max(B))$.
Hence we may define the corresponding inverse function as
$F^{-1}(L,\s) = \min\{\ell ~|~ F(\ell,\s) = L\}$.

Given any interval sequence $\s =(A,B)$, we define
$I(\ell,\s )\triangleq \{i\in[1,n]~|~a_i<\ell<b_i\}$.

\begin{figure}[t]
\centering
\includegraphics[scale=.45]{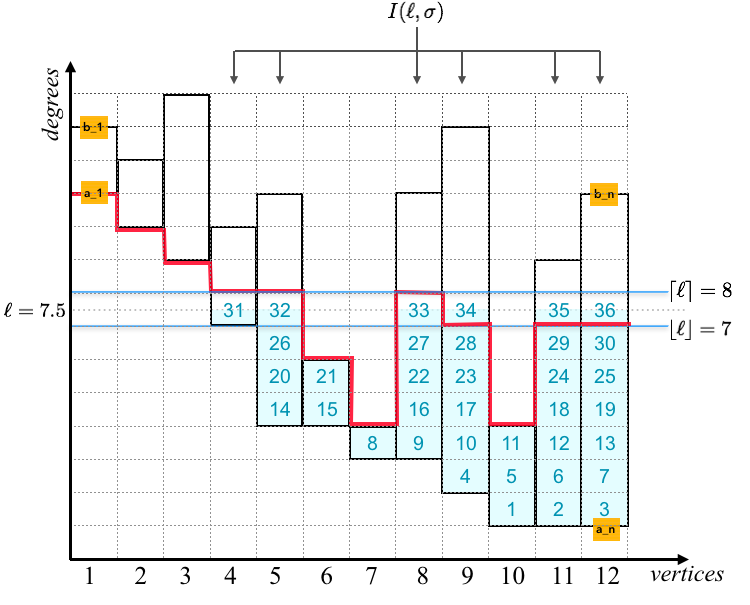}
\caption{Illustration of a levelled sequence $D$ (in red) satisfying
  $L = \vol(D,\s)=33$.  For $\ell=7.5$, $F(\ell=7.5,\s)=33$,
  $F(\floor{\ell} = 7,\s)=30$, and $F(\ceil{\ell} =8,\s)=36$.  The
  segments contributing to $F(\ell=7.5,\s)$, i.e., the parts of the
  connected vessel system filled with fluid, are shown in blue.  The
  values in $D$ at all indices in set $I(\ell,\s)$ differ by at most
  one as they lie in the set $\{\floor{\ell},\ceil{\ell}\}$.  }
\label{Figure:levelled}
\end{figure}

We conclude this section by providing the following theorems for characterising and computing levelled sequences.

\begin{theorem}
\label{theorem:characterisation_levelled_sequence}
Consider an interval sequence $\s = (A,B)$.  Let $L$ be an integer in
$[0, \l1(A,B)]$ and $\ell \geq 0$ be such that $\ell = F^{-1}(L,\s)$.
Then the collection of levelled sequences that have volume $L$ with respect to $\s$ 
is equal to the collection of sequences
$D=(d_1,\ldots,d_n)$ satisfying the following three conditions:
\begin{compactenum}[(a)]
\item $d_i=b_i$ for any $i$ satisfying $b_i \leq \ell$;
\item $d_i=a_i$ for any $i$ satisfying $a_i \geq \ell$;  and
\item Among all indices lying in set $I(\ell,\s)$, exactly $F(\ell,\s)
  - F(\floor{\ell},\s)$ indices $i$ satisfy $d_i = \ceil{\ell}$, and
  the remaining indices $i$ satisfy $d_i = \floor{\ell}$.
\end{compactenum}
\end{theorem}
\begin{proof}
We first make some observations about levelled sequences.
Let $D=(d_1,\ldots,d_n)$ be a levelled sequence with respect to
$\s =(A,B)$ having volume $L$, then
\begin{compactenum}[(i)]
\item For any $i$ with $d_i > \ceil{\ell}$, $d_i=a_i$.
\item For any $i$ with $d_i < \floor{\ell}$, $d_i=b_i$.
\end{compactenum}
To prove claim (i), suppose, to the contrary, that there exists an
index $i_0$ such that $d_{i_0} > \max\{\ceil{\ell},a_{i_0}\}$.  If
there exists some $j_0 \ne i_0$ such that $d_{j_0} < \min\{\ceil{\ell}
,b_{j_0}\}$.  Then one can perform a levelling operation on the pair
$(i_0,j_0)$, to obtain a new sequence $D^*$ that lies between $A$ and
$B$, and has spread $\phi$ strictly less than that of sequence
$D$. This contradicts the fact that $D$ was levelled.  So let us
assume $d_j \ge \min\{\ceil{\ell},b_j\}$ for every $j\ne i_0$. Then
$$
d_j-a_j
~\geq~ \min\{\lceil \ell\rceil ,b_j\}-a_j~ \mbox{for each} ~j\ne i_0
~.
$$
For $i_0$ we have
$$
d_{i_0}-a_{i_0}
~>~    \max\{\ceil{\ell} ,a_{i_0}\}-a_{i_0}
~\geq~ \ceil{\ell} -a_{i_0}
~\geq~ \min\{\ceil{\ell} ,b_{i_0}\}-a_{i_0}
~.
$$
Combining above two inequalities, we get
$$
F(\ell,\s)
~=~ L
~=~ \sum_{1\leq i\leq n} (d_i-a_i)
~>~ \sum_{\substack{1\leq i\leq n,\\~a_i\leq \ceil{\ell}}} 
      (\min\{\ceil{\ell} ,b_i\}-a_i)
~=~ F(\ceil{\ell},\s)
~,
$$
where the inequality follows from the fact that $d_{i_0}-a_{i_0} > 0$.
(This is because if $a_{i_0} < \ceil{\ell}$, then index $i_0$
contributes a positive value of $d_{i_0}-a_{i_0}$ to the first sum and
zero to the second.)  This violates the fact that $F$ is
non-decreasing. The proof of claim (ii) follows in similar manner as
that of claim~(i), and is thus omitted.

We now show that each levelled sequence $D$ w.r.t. $\s $ satisfies
conditions (a), (b) and (c).  We show that for any $i$ satisfying
$b_i\leq \ell$ (or equivalently, $b_i \leq \floor{\ell}$), $d_i =
b_i$.  To prove this, let us assume to the contrary that $d_i < b_i$.
Then $d_i < \floor{\ell}$, which implies $d_i = b_i$, violating our
assumption $d_i < b_i$.  Similarly, it follows that for any $i$
satisfying $a_i \geq \ell$ (or equivalently, $a_i \geq \ceil{\ell}$),
$d_i=a_i$.

Next for any $i\in I(\ell,\s)$, we have
$d_i \in \{\floor{\ell}, \ceil{\ell}\}$.  Indeed, if $d_i
< \floor{\ell}$, then $b_i = d_i < \ell$, violating the fact $\ell <
b_i$.  Similarly, if $d_i > \ceil{\ell}$, then $a_i = d_i > \ell$,
violating the fact $a_i < \ell$.  Since the degrees are integral, the
indices in $I(\ell,\s )$ are partitioned into $L - F(\floor{\ell},\s)$
indices with degree $d_i = \ceil{\ell}$, and $|I(\ell,\s)| - (L -
F(\floor{\ell},\s))$ indices $i$ with degree $d_i = \floor{\ell}$.

To prove the converse, note that all sequences satisfying conditions
(a), (b) and (c) in the theorem are similar and have volume $L$.
\end{proof}

\begin{algorithm}
\caption{\textsc{Levelled-Sequence}($A,B,L$)}
\label{algo:levelled-sequence}
\lIf{$L\notin [0,\l1(A,B)]$}{Return Null}
$A',B' \gets 0$\;
\For{$i=1$ to $n$}{
   $A'[a_i] \gets A'[a_i] + 1$\;
   $B'[b_i] \gets B'[b_i] + 1$\;
}
\lFor{$k=1$ to $n$}{$Q[k] \gets Q[k-1] + A'[k-1] - B'[k-1]$}
$k \gets 1, L' \gets L$\;
\lWhile{$(Q[k] \leq L')$}{$L' \gets L' - Q[k]$ and $k \gets k+1$}
$\ell \gets k - 1 + L'/Q[k]$\;
\For{$i=1$ to $n$}{
   \lIf{$(a_i\geq \ell)$}{$d_i=a_i$}
   \lIf{$(b_i\leq \ell)$}{$d_i=b_i$}
   \If{$(a_i < \ell < b_i)$}{
      \lIf{$L' > 0$}{$d_i = \ceil{\ell}$ and $L' \gets L'-1$}
      \lElse{$d_i =\floor{\ell}$}
   }
}
Return $D$\;
\end{algorithm}

\begin{theorem}
\label{theorem:computing_levelled_sequence}
Given an interval sequence $\s = (A,B)$ consisting of $n$-pairs, and
an integer $L\in [0,\l1(A,B)]$, a levelled sequence $D$ having volume
$L$ with respect to $\s$ can be computed in $O(n)$ time.
\end{theorem}
\begin{proof}
%
For any $k \in [1,n]$, let $Q[k]$ denote the size of the set $\set{i
~|~ a_i < k \leq b_i}$.  In order to efficiently compute $Q$, we
define two $n$-sequences $A'$ and $B'$, where $A'[k]$ and $B'[k]$
count the number of intervals $i$ such that $a_i = k$ and $b_i = k$,
resp.  That is, $A'[k] = |\set{i : a_i = k}|$ and $B'[k] = |\set{i :
b_i = k}|$.  Observe that $Q[k] - Q[k-1] = A'[k-1] - B'[k-1]$, for
every $k \in [1,n]$.


Notice for any integer $\ell$, $F(\ell,\s) = \sum_{k \leq \ell} Q[k]$.
Let $\ell' = \max\{k \in[0,n] ~|~ \sum_{i \leq k} Q[i] \leq L\}$, so
that, $\ell'$ is the maximum integer for which $F(\ell',\s) \leq L$.
Also let
$$
\ell
= \ell' + \frac{L - \sum_{i\leq \ell'}Q[i]}{Q[\ell'+1]}
~.
$$

It is not hard to verify that $\ell = F^{-1}(L,\s )$.  Thus a levelled
sequence of volume $L$ can be computed in $O(n)$ time using
Theorem~\ref{theorem:characterisation_levelled_sequence}.  In
Algorithm~\ref{algo:levelled-sequence} we present the pseudocode of
our implementation that can be seen as filling the vessels with fluid
until reaching the desired level $\ell$.

The correctness of Algorithm~\ref{algo:levelled-sequence} follows from
above description, also it is easy to verify that the total run-time of
the algorithm is $O(n)$.
\end{proof}

\section{An $O(n \log n)$ time algorithm for Graphic Certificate}
\label{section:algo1}

In this section, we present an algorithm for computing a certificate
for interval sequence that takes just $O(n \log n)$ time.  If the
input interval $\s =(A,B)$ is realizable, our algorithm computes a
graphic sequence $D \in \s$, otherwise it computes a sequence
minimizing the deviation value $\delta(D,\s)$.
We begin by considering the case where the sequence $\s$ is
realizable~(since it is simpler to understand given Theorems~\ref{theorem:characterisation_levelled_sequence} and
\ref{theorem:computing_levelled_sequence}), and then we move to the case where $\s$ is non-realizable.
Then characterization of \cite{CDZ:00} implies 
an $O(n)$ time verification algorithm for realizability of interval sequence.  
(For details refer to the Appendix).


\subsection{Realizable Interval Sequences}

First we show that any two levelled sequences after an appropriate
reordering of their elements are coordinate-wise comparable.

\begin{lemma}
\label{lemma:similar_sequence}
For any interval sequence $\s =(A,B)$, and any two levelled sequences
$C, D \in \s$ satisfying $\vol(D,\s) \leq \vol(C,\s)$, the following
holds.
\begin{compactenum}
\item $D' \leq C$, for some sequence $D' \in \s$ similar to $D$.
\item $D \leq C''$, for some sequence $C'' \in \s$ similar to $C$.
\end{compactenum}
\end{lemma}
\begin{proof}
We show how to transform $D=(d_1,\cdots,d_n)$ into sequence $D' =
(d'_1,\cdots,d'_n) \in \s$ such that $D' \leq C$.  Let $\ell_D =
F^{-1}(\vol(D,\s),\s)$ and $\ell_C = F^{-1}(\vol(C,\s),\s)$.  Since
$F(\cdot,\s)$ is a non-decreasing function, we have that $\ell_D \leq
\ell_C$.

Let us first consider the case where $\ell_C$ and $\ell_D$ are both
non-integral, and $\floor{\ell_C} = \floor{\ell_D}=(\text{say }\ell_1)$ and
$\ceil{\ell_C} = \ceil{\ell_D}=(\text{say }\ell_2)$.
By Theorem~\ref{theorem:characterisation_levelled_sequence}, for any
index $i \in [1,n]$,
\begin{inparaenum}[(i)]
\item $a_i\geq \ell_D$ (or $a_i\geq \ell_C$) implies $d_i=a_i=c_i$;
\item $b_i\leq \ell_D$ (or $b_i\leq \ell_C$) implies $d_i=b_i=c_i$.
\end{inparaenum}
Also, among indices in set $I_0 = I(\ell_D,\s) = I(\ell_C,\s)$,
\begin{inparaenum}[(i)]
\item exactly $L_D - F(\floor{\ell_D},\s)$ indices $i$ satisfy $d_i =
  \ell_2$ (let $I_D$ denote the set of these indices) and the
  remaining indices $i$ satisfy $d_i = \ell_1$;
\item exactly $L_C - F(\floor{\ell_C},\s)$ indices $i$ satisfy $c_i =
  \ell_2$ (let $I_C$ denote the set of these indices) and the
  remaining indices $i$ satisfy $c_i = \ell_1$.
\end{inparaenum}
Since $L_D\leq L_C$, it follows that $|I_D|\leq |I_C|$, however,
observe that $I_D$ need not be a subset of $I_C$.  We set $D'$ to be
the sequence
that satisfy the condition that
\begin{inparaenum}[(i)]
\item $d'_i=d_i$, for each $i\notin I_0$, and
\item for indices in $I_0$, at any arbitrary $|I_D|$ indices lying in
  $I_C$, $d'_i$ take the value $\ell_2$, and at remaining
  $|I_0|-|I_D|$ indices $d'_i$ take the value $\ell_1$.
\end{inparaenum}
It is easy to verify that $D$ and $D'$ are similar, and
$D'\leq C$.

The remaining case is when $\ceil{\ell_D} \leq \floor{\ell_C}$.  For
any index $i \in [1,n]$, $d_i \leq \ceil{\ell_D}$ and $c_i \geq
\floor{\ell_C}$, implies $d_i \leq c_i$.  Observe that by
Theorem~\ref{theorem:characterisation_levelled_sequence},
\begin{inparaenum}[(i)]
\item for an index $i$, $d_i>\lceil\ell_D\rceil$ implies $d_i=a_i(\leq
  c_i)$; and
\item for an index $i$, $c_i<\lfloor\ell_C\rfloor$ implies
  $c_i=b_i(\geq d_i)$.
\end{inparaenum}
Therefore, for each index $i$, $d_i\leq c_i$.  So in this case, we set
$D'$ to be $D$.
The construction of sequence $C''$ follows similarly.
\end{proof}

Next lemma shows significance of partitioning  an interval-sequence using a levelled sequence.

\begin{lemma}
\label{lemma:subdivide_interval}
Let $C$ and $D$ be any two levelled sequences lying in an interval
sequence $\s = (A,B)$, and having volume $L_C$ and $L_D$,
respectively. Also assume $D$ is a graphic sequence. Then,
\begin{compactenum}[(a)]
\item $L_D\leq L_C$ implies $(A,C)$ is a realizable interval sequence.

\item $L_D\geq L_C$ implies $(C,B)$ is a realizable interval sequence.
\end{compactenum}
\end{lemma}

\begin{proof}
We provide proof of the case $L_D \leq L_C$ (the proof of part (b)
will follow in a similar fashion).  By
Lemma~\ref{lemma:similar_sequence}, we can transform $D =
(d_1,\cdots,d_n)$ into another levelled sequence $D' =
(d'_1,\cdots,d'_n) \in \s$ such that $D'$ is similar to $D$ and $D'
\leq C$.  Since $D' \leq C$, and $D'$ is a graphic sequence, it
follows that $(A,C)$ is realizable interval sequence.
\end{proof}

From Lemma~\ref{lemma:subdivide_interval}, and the fact that each
realizable interval-sequence contains a levelled graphic sequence (see
Theorem~\ref{theorem:levelled_graphic_sequence}), we obtain following.

\begin{theorem}
For any realizable interval sequence $\s =(A,B)$, and any levelled
sequence $C \in \s$, 
at least one of the interval-sequences $(A,C)$ and $(C,B)$ is realizable.
\label{theorem:subdivide_interval}
\end{theorem}

The above theorem provides a divide-and-conquer strategy to search for
a levelled graphic sequence for realizable interval-sequences
as shown in Algorithm~\ref{algo:binary-search}.
Let $(A_0,B_0)$ be initialized to $(A,B)$.  We compute a levelled
sequence $C_0$ having volume $\floor{\l1(A_0,B_0)/2}$ using
Theorem~\ref{theorem:characterisation_levelled_sequence}.  It follows
from Theorem~\ref{theorem:subdivide_interval}, either $(A_0,C_0)$ or
$(C_0,B_0)$ must be a realizable interval-sequence. If $(A_0,C_0)$ is
realizable then we replace $B_0$ by $C_0$; otherwise $(C_0,B_0)$ must
be realizable, so we replace $A_0$ by $C_0$. We continue this process
(of replacements) until $\l1(A_0,B_0)$ decreases to a value smaller
than $2$. In the end, the interval sequence $(A_0,B_0)$ contains at
most two sequences, namely $A_0$ and $B_0$.  If $A_0$ is graphic then
we return $A_0$, otherwise we return $B_0$. The correctness of the
algorithm is immediate from the description.
It is also easy to verify that the algorithms outputs a graphic-certificate 
with the least possible number of edges.

\begin{algorithm}
\caption{\textsc{Certificate-Realizable}($A,B$)}
\label{algo:binary-search}
Initialize interval sequence $(A_0,B_0)$ to $(A,B)$\;
\While{$\l1(A_0,B_0)\geq 2$}{
$C_0\gets $ a levelled sequence of volume $\floor{\l1(A_0,B_0)/2}$\;
\lIf{\textup{(Interval-sequence $(A_0,C_0)$ is realizable)}}
	{$B_0\gets C_0$}
\lElse {$A_0\gets C_0$}
}
\lIf{$A_0$ is graphic}{Return $A_0$}
\lElse{Return $B_0$}
\end{algorithm}

To analyze the running time, observe that the $\l1$-distance
between $A_0$ and $B_0$ decreases by (roughly) a factor of $2$ in
each call of the while loop, so it follows that number of iterations is~
$O(\log n)$.  Verifying if an interval sequence is realizable, or a
sequence $D$ is graphic can be performed in $O(n)$ time, using
Theorem~\ref{theorem:implementation}.  Also in $O(n)$ time we can
generate a levelled sequence of any given volume $L$ by
Theorem~\ref{theorem:computing_levelled_sequence}.  Thus, the total
time complexity of the algorithm is $O(n\log n)$.

We obtain the following result:

\begin{theorem}
There exists an algorithm that for any integer $n \geq 1$ and any
$n$-length interval sequence $\s =(A,B)$, computes a graphic sequence
$D \in (A,B)$, if it exists, in $O(n\log n)$ time.

Moreover, our algorithm can also output graphic certificates corresponding to a sparsest and densest\footnote{Note that the sparsest (resp. densest) possible graph realizing an interval sequence $\s$ need not be unique.} possible graph (i.e. having minimum and maximum possible edges), 
realizing $\s$, in the same time.
\label{theorem:algo1_realizable}
\end{theorem}


\subsection{Non-Realizable Sequences}

In this subsection we consider the scenario where $\s$ is non-realizable,
our goal is to compute a graphic sequence $D$ minimizing the
deviation $\delta(D,\s)$ with respect to the given interval sequence $\s$.

As a first step, we show that in order to search a sequence $D$
minimizing $\delta(D,\s)$, it suffices to search a sequence $D \geq A$
that minimizes the value $\delta_U(D,\s)$.

\begin{lemma}
\label{lemma:upper_range_suffices}
$\min\{\delta(D,\s )~|~D\text{ is graphic }\}
= \min\{\delta_U(D,\s )~|~D\text{ is graphic, } D\geq A\}$,
for any interval sequence $\s =(A,B)$.
\end{lemma}
\begin{proof}
Let $D = (d_1,\ldots,d_n)$ be a graphic sequence minimizing the value
$\delta(D,\s)$, and in case of ties take that $D$ for which
$\delta_L(D,\s)$ is the lowest.  Let us suppose there exists an index
$i \in [1,n]$ such that $d_i < a_i$.  Consider the graph $G$ realizing
the sequence $D$, and let $v_i$ denote the $i$th vertex of $G$, so
that, $\deg(v_i) = d_i$.
Observe that $|N_G(v_i)| \neq n-1$, since $d_i < a_i \leq n-1$.
For any vertex $v_j \notin N_G(v_i)$, $d_j = \deg(v_j)$
must be at least $b_j$, because otherwise adding $(v_i,v_j)$ to $G$
reduces $\delta(D,\s)$.  Thus for any vertex $v_j \notin N_G(v_i)$,
adding $(v_i,v_j)$ to $G$, decreases $\delta_L(D,\s )$ and increases
$\delta_U(D,\s)$ by a value exactly $1$.  However, by our choice $D$
was a sequence minimizing $\delta_L(D,\s )$, thus $\delta_L(D,\s )$
must be zero.  The claim follows from the fact that $D\geq A$ and
$\delta(D,\s) = \delta_U(D,\s)$.
\end{proof}

By the previous lemma, our goal is to find a graphic sequence $D$ in
the interval sequence $(A,\top)$ minimizing $\delta(D,\s)$.  Notice
that if $D$ is graphic, then the interval sequence $(A,R)$, where $R =
\max(D,B)$, is realizable.  Also, $\delta(D,S) = \textsum(R - B)$.
Hence, in order to compute a graphic sequence with minimum deviation,
we define $\R$ to be the set of all sequence $R \in [B,\top]$ such
that~
\begin{inparaenum}[(i)]
\item the interval sequence $(A,R)$ is realizable, and
\item $\textsum(R-B)$ is minimized.
\end{inparaenum}

The following lemma shows the significance of the set $\R$ in computing a
certificate with minimum deviation.

\begin{lemma}
\label{lemma:R_significance}
For any $R \in \R$, and any graphic sequence $D_0$ lying in the
interval sequence $(A,R)$, we have $\delta(D_0,\s) =
\min\{\delta(D,\s)~|~D\text{ is graphic }\}=\textsum(R-B)$.
\end{lemma}
\begin{proof}
Let $D^*$ be a graphic sequence minimizing the value $\delta(D,\s)$.
By Lemma~\ref{lemma:upper_range_suffices}, we may assume that $D^*$
belongs to $(A,\top)$.
Observe that $\delta(D^*,\s) = \textsum(R^*-B)$, where $R^* =
\max\{B,D^*\}$.
By the choice of $D^*$ we have that
$
\textsum(R^* - B)
=    \delta(D^*,\s)
\leq \delta(D_0,\s)
=    \delta_U(D_0,\s)
\leq \textsum(R - B)
$,
where the last inequality follows from the fact that $D_0 \in (A,R)$.
By definition of $\R$, we have that $\textsum(R^* - B) \geq \textsum(R
- B)$, and therefore $\delta(D^*,\s) = \delta(D_0,\s) = \textsum(R -
B) = \textsum(R^* - B)$.  Thus $R^*$ also lies in the set $\R$.  The
lemma follows from the fact that $\delta(D^*,\s) =
\min\{\delta(D,\s)~|~D\text{ is graphic }\}$.
\end{proof}

Next, let $\R_L$ be the set of all levelled sequences in $\R$ with
respect to interval sequence $(B,\top)$.

\begin{lemma}
$\R_L \neq \emptyset$.
\end{lemma}
\begin{proof}
Clearly, $\R \neq \emptyset$.  Consider any sequence $R =
(r_1,\ldots,r_n) \in \R$.
Suppose there exists $\ia, \ib \in[1,n]$ such that $r_\ia - r_\ib \geq
1$ and $R' = \pi(R,\ia,\ib) \in [B,\top]$.
Observe that $\textsum(R'-B)=\textsum(R-B)$.  It remains to show that
$(A,R')$ is realizable.
Indeed, if $D \in (A,R)$ is a graphic sequence, then either
\begin{inparaenum}[(i)]
\item $D = (d_1,\ldots,d_n)$ lies in $(A,R')$, or
\item $d_\ia - d_\ib \geq 1$ and $D' = \pi(D,\ia,\ib)$ lies in $(A,R')$.
\end{inparaenum}
  Since levelling
operation preserves graphicity, $D'$ is graphic. Thus $R'\in \R$,
which shows that $\R$ is closed under the levelling operation, and
hence $\R_L$ is non-empty.
\end{proof}

\begin{algorithm}
\caption{\textsc{Certificate-Non-Realizable}($A,B$)}
\label{algo:binary-search-non}
$(M_1,M_2)\gets (B,\top)$\;
\While{$\l1(M_1,M_2)\geq 2$}{
$M_0\gets $ a levelled sequence of volume $\floor{\l1(M_1,M_2)/2}$\;
\lIf{\textup{(Interval-sequence $(A,M_0)$ is realizable)}}
	{$M_2\gets M_0$}
\lElse {$M_1\gets M_0$}
}
\lIf{$(A,M_1)$ is realizable}{$R\gets M_1$}
\lElse{$R\gets M_2$}
Return \textsc{Certificate-Realizable}($A,R$)
\end{algorithm}

We now describe the algorithm for computing a graphic sequence with minimum deviation
(refer to Algorithm~\ref{algo:binary-search-non} for a pseudocode). 
Recall that we assume that $(A,B)$ is a non-realizable interval sequence.
The first step is to compute a levelled sequence $R \in \R_L$, and the
second is to use Algorithm~\ref{algo:binary-search} to find a graphic
sequence in $(A,R)$.

We initialize two sequences $M_1$ and $M_2$,
resp., to $B$ and $\top$, and these sequences serve as lower
and upper boundaries for sequence $R$.  The pair $(M_1,M_2)$ is
updated as long as $\textsum(M_2-M_1) \geq 2$ as follows.  We compute
a levelled sequence $M_0$ having volume $\floor{\l1(M_1,M_2)/2}$ with
respect to the interval sequence $(M_1,M_2)$ using
Theorem~\ref{theorem:computing_levelled_sequence}.  There are two
cases:
\begin{description}
\item[{\bf Case 1.}] $(A,M_0)$ is realizable.

  Consider any sequence $R \in(M_1,M_2)$ that lies in $\R_L$.  Since
  $(A,M_0)$ is realizable, from the definition of $R$ it follows that
  $\textsum(R-B) \leq \textsum(M_0-B)$.  As $R$ and $M_0$ both belong
  to $(M_1,M_2)$, by Lemma~\ref{lemma:similar_sequence}, there exists
  a sequence $R_0$ similar to $R$ lying in interval $(M_1,M_2)
  \subseteq (B,\top)$ such that $R_0 \leq M_0$.  It is easy to check
  that $R_0 \in \R_L$, thus the search range of $R$ which was
  $(M_1,M_2)$ can be narrowed down to $(M_1,M_0)$, so we reset $M_2$
  to $M_0$.

\item[{\bf Case 2.}] $(A,M_0)$ is not realizable.
  
  Consider any $R\in \R_L$, we first show that $\textsum(R-B) >
  \textsum(M_0-B)$.  Let us assume on the contrary, $\textsum(R-B)\leq
  \textsum(M_0-B)$. In such a case, by
  Lemma~\ref{lemma:similar_sequence}, there exists a sequence $R'$
  similar to $R$ lying in interval $(M_1,M_2)\subseteq (B,\top)$ such
  that $R'\leq M_0$.  Also $R'\in \R_L$. Since, by definition of
  $\R_L$, $(A,R')$ is realizable, it violates the fact that $(A,M_0)$
  is not realizable.  Now as $R,M_0$ both belong to $(M_1,M_2)$, by
  Lemma~\ref{lemma:similar_sequence}, there exists a sequence $R_0$
  similar to $R$ lying in interval $(M_1,M_2)\subseteq (B,\top)$ such
  that $R_0\geq M_0$. Also $R_0\in \R_L$, thus the search range of $R$
  can be narrowed down to $(M_0,M_2)$, so we reset $M_1$ to $M_0$.
\end{description}

We continue the process of shrinking the range $(M_1,M_2)$ until
$\l1(M_1,M_2)$ decreases to a value smaller than $2$. Finally there
exists in range $(M_1,M_2)$ at most two sequences, namely $M_1$ and
$M_2$. If $(A,M_1)$ is graphic then we set $R$ to $M_1$, otherwise we
set $R$ to $M_2$.

The running time analysis is similar to the one for
Algorithm~\ref{algo:binary-search}.  Since the $\l1$-distance between
$M_1$ and $M_2$ decreases by a factor of $2$ in each successive call
of the while loop of the algorithm, it follows that number of times
the while loops run is $O(\log n)$.  Verifying if an interval sequence
is realizable, or a sequence $D$ is graphic can be performed in $O(n)$
time, using Theorem~\ref{theorem:implementation}.  Also it takes
$O(n)$ time to generate a levelled sequence of any given volume $L$ by
Theorem~\ref{theorem:computing_levelled_sequence}.
Finally, the running time of Algorithm~\ref{algo:binary-search} is
$O(n \log n)$.  Thus, the total time complexity of algorithm is
$O(n\log n)$.

This completes the proof of Theorem~\ref{theorem-our-result-1}.


\section{Most Regular Certificate in $O(n^2)$ time}
\label{section:algo2}

In this section, we present an $O(n^2)$-time algorithm for computing a
most-regular certificate with respect to a given interval sequence $\s =
(A,B)$.  We assume that $\s$ is realizable. 
%
Our algorithm involves a subroutine that given an integer $z
\in [\min(A),\max(B)-1]$, computes a most-regular graphic-sequence,
say $D$, satisfying the condition $z \leq \ell = F^{-1}(\vol(D,\s),\s)
\leq z+1$.
The following lemma is immediate from
Theorem~\ref{theorem:characterisation_levelled_sequence}.

\begin{lemma}
\label{observation-z}
Any levelled sequence $\bar{D} = (\bar d_1,\ldots,\bar d_n)$ of volume
$L$ with respect to interval sequence $\s = (A,B)$, satisfies $z \leq
\ell = F^{-1}(L,\s) \leq z+1$ if and only if $\bar d_i= a_i$ for $a_i
\geq z+1$, $\bar d_i=b_i$ for $b_i \leq z$, and $\bar d_i \in
\{z,z+1\}$ for remaining indices $i$.
\end{lemma}

We partition the set $[1,n]$ into three sets $I_1$, $I_2$, $I_3$
such that $I_1=\{i\in[1,n]~|~a_i\geq z+1\}$, 
$I_2=\{i\in[1,n]~|~a_i\leq z~\mbox{ and }~z+1\leq b_i\}$, and
$I_3=\{i\in[1,n]~|~b_i\leq z\}$.
Also, using integer sort in linear time, we rearrange the pairs in
$(A,B)$ along with the corresponding sets $I_1,I_2,I_3$ so that
\begin{inparaenum}[(i)]
\item for any $i \in I_1$, $j \in I_2$, $k \in I_3$, we have $i < j <
  k$, and
\item the sub-sequences $A[I_1]$ and $B[I_3]$ are sorted in the
  non-increasing order.
\end{inparaenum}

We initialize $D_z = (d_{z,1},d_{z,2},\ldots,d_{z,n})$ by setting
$d_{z,i}$ to : $a_i$~if~$i\in I_1$, $z$ if $i\in I_2$, and $b_i$ if $i\in I_3$.
%
The sequence $D_z$ is sorted in non-increasing order, since the
sub-sequences $A[I_1]$ and $B[I_3]$ are sorted in non-increasing
order.  Let $\alpha = |I_1|$ and $\beta = |I_1|+|I_2|$, so that $I_2 =
[\alpha+1,\alpha+2,\ldots,\beta]$.
We would search all those indices $i\in [\alpha,\beta]$ such that on
incrementing $d_{\alpha+1}, \ldots, d_{i}$ to value $z+1$, the
resulting sequence is graphic; or equivalently, the sequence $D_z +
E_{[\alpha+1,i]}$ is graphic.  Note that for any index
$i\in[\alpha,\beta]$,
\begin{inparaenum}[(i)]
\item the sequence $D_z+E_{[\alpha+1,i]}$ is non-increasing, and
\item $A \leq D_z+E_{[\alpha+1,i]}\leq B$.
\end{inparaenum}
The next lemma, which follows from the definition of $\phi$, will be
used to compute $\phi(D_z+E_{[\alpha+1,i]})$ from $\phi(D_z)$.

\begin{lemma}
\label{lemma:phi_relation}
For any index $i \in[\alpha+1,\beta]$,
$\phi(D_z+E_{[\alpha+1,i]})=\phi(D_z)+(i-\alpha)(n-i-\alpha)$.
\end{lemma}
\begin{proof}
Let $D^i_z = D_z + E_{[\alpha+1,i]}$, for $i \in [\alpha+1,\beta]$.
\begin{align*}
\phi(D^i_z)
= & \sum_{j=1}^{n-1} \sum_{k=j+1}^n \abs{d^i_{z,j} - d^i_{z,k}} \\
= & \sum_{j \leq \alpha}
      \paren{
         \sum_{k > j, k \not\in [\alpha+1,i]} \abs{d_{z,j} - d_{z,k}} +
         \sum_{k \in [\alpha+1,i]} \abs{d_{z,j} - d_{z,k} - 1}
      } + \\
  & \sum_{j \in [\alpha+1,i]}
      \paren{
         \sum_{k \in [j+1,i]} \abs{d_{z,j} - d_{z,k}} + 
         \sum_{k > i} \abs{d_{z,j} + 1 - d_{z,k}} +
      } + \\
  & \sum_{j > i} \sum_{k > j} \abs{d_{z,j} - d_{z,k}} \\
= & \sum_{j=1}^{n-1} \sum_{k=j+1}^n \abs{d_{z,j} - d_{z,k}} 
    - \alpha \cdot (i-\alpha) + (i-\alpha) \cdot (n-i) \\
= & \phi(D_z)+ (i - \alpha) \cdot (n - i - \alpha)
~,
\end{align*}
and the lemma follows.
\end{proof}

For each $z$ we compute the vectors $X(D_z)$ and $Y(D_z)$ using
Theorem~\ref{theorem:implementation} in the Appendix.  For each
integer $k\in [1,n]$, let
$$\avd(k) = \set{i\in[\alpha,\beta]~\mid~ X_k\big(D_z+E_{[\alpha+1,i]}\big) >
                                  Y_k\big(D_z+E_{[\alpha+1,i]}\big) }\textup{,~and}$$
$$\avd = \textstyle \bigcup_{k=1}^n \avd(k)~.$$

Observe that by Theorem~\ref{theorem:EG}, for any $i \in
[\alpha,\beta]$, the sequence $D_z + E_{[\alpha+1,i]}$ is graphic if
and only if $i$ does not lie in the set $\avd$, and
$\parity(D_z+E_{[\alpha+1,i]}) = 0$.
The following lemma shows that the set  $\avd(k)$, for any index $k$,
is computable in $O(1)$ time.

\begin{lemma}
\label{lemma:avd_contiguous}
For each $k \in [1,n]$, $\avd(k)$ is a contiguous sub-interval of $[1,n]$,
and is computable in $O(1)$ time.
\end{lemma}

\begin{proof}
For each $k \in [1,n]$, let
\begin{align*}
\avd_1(k)
& = [\alpha,\beta] \cap [\alpha + 1+Y_k(D_z)-X_k(D_z),k] \\
\avd_2(k)
& = [\alpha,\beta] \cap
\begin{cases}
[k,n] & \hspace{-35mm} k \leq z \text{ and } \max\{k-\alpha,0\}+X_k(D_z) > Y_k(D_z);\\
[k,\max\{k,\alpha\}+X_k(D_z)-Y_k(D_z)-1] & \hspace{20mm} k \geq z+1;\\
\emptyset  & \hspace{-19mm} \max\{k-\alpha,0\}+X_k(D_z)\leq Y_k(D_z).
\end{cases}
\end{align*}

We fix an index $k\in[1,n]$ for the rest of the proof.
Our goal will be to show that the
set $\avd(k)$ is the union of the interval $\avd_1(k)$
and $\avd_2(k)$, and is thus computable in $O(1)$ time.
Consider any $i \in [\alpha,\beta] \cap [1,k]$.  Since $\alpha \leq i
\leq k$, we have that
\begin{align*}
X_k(D_z + E_{[\alpha+1,i]})
& = \textstyle \sum_{j=1}^k d_{z,j} + (i-\alpha)
  = X_k(D_z)+(i-\alpha)
~, \\
Y_k(D_z + E_{[\alpha+1,i]})
& = \textstyle k(k-1) + \sum_{j = k+1}^n \min(d_{z,j},k)
  = Y_k(D_z)
~.
\end{align*}
Observe $i \in \avd(k)$ if and only if $X_k(D_z+E_{[\alpha+1,i]}) \geq
Y_k(D_z+E_{[\alpha+1,i]}) + 1$, which is equivalent to
$$
\alpha + 1 + Y_k(D_z) - X_k(D_z) \leq i
~.
$$
It follows that
$$
\avd(k) \cap [1,k]
= [\alpha,\beta] \cap [\alpha + 1+Y_k(D_z)-X_k(D_z),k]
= \avd_1(k)
~.
$$
Next consider any $i \in [\alpha, \beta] \cap [k,n]$.  As $k \leq i
\leq \beta$, we have that
\begin{align*}
X_k(D_z+E_{[\alpha+1,i]})
& = \textstyle \sum_{j=1}^k d_{z,j} + \max\set{0,k-\alpha}
  = X_k(D_z) + \max\set{0,k-\alpha}
~,
\\
Y_k(D_z+E_{[\alpha+1,i]})
& =
\begin{cases}
Y_k(D_z)  &  k\leq z,
\\
Y_k(D_z)+(i-\alpha)  &  k\geq z+1.
\end{cases}
\end{align*}
We have the following three different cases:
\begin{enumerate}[\text{~~Case} 1:]
\item $k \leq z$ and $\max\{0,k-\alpha\}+X_k(D_z) > Y_k(D_z)$:

  In this case $X_k(D_z+E_{[\alpha+1,i]}) = X_k(D_z) +
  \max\{0,k-\alpha\}$ and $Y_k(D_z+E_{[\alpha+1,i]}) = Y_k(D_z)$,
  hence $X_k(D_z+E_{[\alpha+1,i]}) > Y_k(D_z+E_{[\alpha+1,i]})$,
  implying that $i\in \avd_2(k)$.  Therefore, in this case $\avd_2(k)
  = [\alpha,\beta] \cap [k,n]$.

\item $\max\{0,k-\alpha\} + X_k(D_z) \leq Y_k(D_z)$:

  In this case $X_k(D_z+E_{[\alpha+1,i]}) = X_k(D_z) +
  \max\{0,k-\alpha\} \leq Y_k(D_z)\leq Y_k(D_z+E_{[\alpha+1,i]})$,
  thus $\avd_2(k)=\emptyset$.

\item $k \geq z+1$:
  
  In this case $Y_k(D_z+E_{[\alpha+1,i]}) = Y_k(D_z)+(i-\alpha)$.
  Observe that, $i \in \avd(k)$ if and only if
  $X_k(D_z+E_{[\alpha+1,i]})\geq Y_k(D_z+E_{[\alpha+1,i]}) + 1$.  The
  latter is satisfied if and only if $X_k(D_z) + \max\set{0,k-\alpha}
  \geq Y_k(D_z) + (i-\alpha) + 1$, or if
$$
i \leq \max\set{\alpha,k} + X_k(D_z) - Y_k(D_z) -1 ~.
$$
  Hence, $\avd_2(k) = [\alpha,\beta] \cap
  [k,\max\{k,\alpha\}+X_k(D_z)-Y_k(D_z)-1]$.
\end{enumerate}
The lemma follows.
\end{proof}

\begin{algorithm}[!ht]
\caption{$\textsc{Most}$-$\textsc{Regular}$-$\textsc{Certificate}(A,B)$}
\label{Algorithm:regular-certificate}
$\textsc{opt} \gets \infty$\;
\ForEach{$z\in [\min(A),\max(B)-1]$}{
$I_1\gets \{i\in[1,n]~|~a_i\geq z+1\}$\;
$I_2\gets \{i\in[1,n]~|~a_i\leq z, z+1\leq b_i\}$\;
$I_3\gets \{i\in[1,n]~|~b_i\leq z\}$\;
Rearrange the pairs in $(A,B)$ along with the corresponding sets $I_1,I_2,I_3$ so that
\quad (i)~for~any triplet $(i,j,k)$ satisfying $i\in I_1$, $j\in I_2$, $k\in I_3$, we have $i<j<k$, and
\quad (ii)~the sub-sequences $A[I_1]$ and $B[I_3]$ are sorted in the non-increasing order\;
Initialize $D_z=(d_1,d_2,\ldots,d_n)$, where for $i\in I_1$, $d_i=a_i$; for $i\in I_2$, $d_i=z$; 
		and for $i\in I_3$, $d_i=b_i$\;
\lIf{$z=\min(A)$}{set $\phi(D_z)=\sum_{1\leq r<s\leq n} |a_r-a_s|$}
Compute $X(D_z),Y(D_z)$ using Theorem~\ref{theorem:implementation}\;
Let $\alpha=|I_1|$ and $\beta=|I_1|+|I_2|$\;
\For{$k=1$ to $n$}{
{$\avd_1(k)= [\alpha,\beta] \cap [\alpha+1+Y_k(D_z)-X_k(D_z),k]$\;}
\vspace{.5mm}
\lIf{$\max\{k-\alpha,0\}+X_k(D_z) > Y_k(D_z)$ and $k \leq z$}
		{$\avd_2(k)=[\alpha,\beta] \cap [k,n]$}
\lElseIf{$\max\{k-\alpha,0\}+X_k(D_z)\leq Y_k(D_z)$}{$\avd_2(k)=\emptyset$} 
\lElse{$\avd_2(k)=[\alpha,\beta] \cap [k,\max\{k,\alpha\}+X_k(D_z)-Y_k(D_z)-1]$}
\vspace{.5mm}
$\avd(k)=\avd_1(k) \cup \avd_2(k)$\;
}
Compute $\avd=\bigcup_{k=1}^n \avd(k)$\;
\vspace{.5mm}
\ForEach{$i\in~[\alpha,\beta]\setminus \avd$}{
\If{$(\parity(D_z)= (i-\alpha) \mod 2)$}{
	Compute $\phi(D_z+E_{[\alpha+1,i]})=\phi(D_z)+(i-\alpha)(n-i-\alpha)$\;
	$\textsc{opt}=\min\{\textsc{opt},\phi(D_z+E_{[\alpha+1,i]})\}$\;
	}
}
Set $\phi(D_{z+1})=\phi(D_z)+(\beta-\alpha)(n-\beta-\alpha)$\;
}
Return $\textsc{opt}$ and the corresponding graphic sequence\;
\end{algorithm}

Algorithm~\ref{Algorithm:regular-certificate} presents the procedure
for computing the most-regular certificate.
For each $k\in[1,n]$, $\avd(k)$ is a contiguous
sub-interval of $[1,n]$, therefore, the union $\avd = \bigcup_{k=1}^n
\avd(k)$ can be computed in linear time using simple stack based
data-structure, once the intervals are sorted in order of their
endpoints%
\footnote{We say $[r,s]\leq [r',s']$ if either
\begin{inparaenum}[(i)]
\item $r<r'$, or
\item $r=r'$ and $s \leq s'$.
\end{inparaenum}
}
using integer sort.
Let ${\cal I}_z$ denote the set obtained by removing from
$[\alpha,\beta] \setminus \avd$ each index $i$ for which $\parity(D_z
+ E_{[\alpha+1,i]}) = \parity(\textsum(D_z) + (i-\alpha))$ is
non-zero.  Since $\textsum(D_z)$ (or $\parity(D_z)$) is computable in
$O(n)$, the set ${\cal I}_z$ can be computed in $O(n)$ time as well.
Note that $D_z + E_{[\alpha+1,i]}$ is graphic if and only if $i \in
{\cal I}_z$.  By Lemma~\ref{lemma:phi_relation}, for any index $i \in
{\cal I}_z$, the value $\phi(D_z+E_{[\alpha+1,i]})$ is computable in
$O(1)$ time, once we know $\phi(D_z)$.  This shows that in just $O(n)$
time, we can compute the spread of all the levelled sequences $D$
satisfying $z \leq F^{-1}(\vol(D),\s) < z+1$, and also find a sequence
having the minimum spread.  All~that remains is to efficiently
computing $\phi(D_z)$ for each $z\in[\min(A),\max(B)]$.  Observe that
$D_{\min(A)}=A$, and so $\phi(D_{\min(A)})=\sum_{1\leq r<s\leq n}
|a_r-a_s|$ is computable in $O(n^2)$ time. Next by
Lemma~\ref{lemma:phi_relation}, for any $z\in [\min(A),\max(B)-1]$,
$\phi(D_{z+1})=\phi(D_z)+(\beta-\alpha)(n-\beta-\alpha)$ is computable
in $O(1)$ time.  Since $z$ can take $\max(B)-\min(A)-1$ values, our
algorithm in total takes $O(n^2~+~n(\max(B)-\min(A)-1))=O(n^2)$
time.

This completes the proof of Theorem~\ref{theorem-our-result-2}.


\section{Applications and Extensions}
\label{section:application}

In this section, we discuss some related problems whose solutions follow as immediate application
of our interval sequence work.

\begin{problem}[\normalfont{Minimum Graphic extensions}]
Given a sequence $A=(a_1,\ldots,a_p)$ 
find the minimum integer $n(\geq p)$ 
such that a super sequence $D=(a_1,\ldots,a_p,d_{p+1},d_{p+2},\ldots,d_n)$ of sequence $A$ is realizable.
\end{problem}
{\em Solution}:
Let $M$ denote the value $\max(A)=\max_{i\in[1,p]} a_i$.
For any $n\geq p$, let $\s_n =([a_1,a_1],\ldots,$
$[a_p,a_p],[1,n],\ldots,[1,n])$ denote the sequence obtained by appending 
$n-p$ copies of interval $[1,n]$ to interval sequence $(A,A )$.
Let $n_0$ the denote the length of a minimum graphic extension of~$A$.
Observe that $n_0\in [\max\{p,M\},p+M]$. 
The lower limit is due to the fact that the length of minimum graphic
extension of $A$ must be at least $\max\{p,M\}$;
the upper limit holds since one can have a bipartite graph with partitions $X=\{x_1,\ldots,x_p\}$
and $Y=\{y_1,\ldots,Y_M\}$ of length $p$ and $M$, and for $i\in[1,p]$, connect the vertex $x_i$ 
to vertices $y_1,\ldots,y_{a_i}$.
It turns out that we need to find the smallest integer $n\in[\max\{p,M\},p+M]$ such that
$\s_n$ is graphic.
The minimum $n$ can be obtained by a binary search over the range $[\max\{p,M\},p+d]$ and 
using Theorem~\ref{theorem:CDZ}; this takes $O(\max\{p,M\} \log \max\{p,M\})$ time.
Once $n_0$ is known, the optimal graphic extension can be computed using Theorem~\ref{theorem:algo1_realizable}
for searching graphic certificate 
in $O(\max\{p,M\} \log \max\{p,M\})$ time.


\begin{problem}
Given $A=(a_1,\ldots,a_n)$, find a graphic sequence $D=(d_1,\ldots,d_n)$ whose
chebyshev distance ($L_\infty$ distance) from $A$ is minimum.
\end{problem}
{\em Solution}: The above problem can be reduced to interval sequence problem,
as we need to find smallest non-negative integer
$c\in[1,n]$ such that $\s_c =([a_1-c,a_1+c],\ldots,[a_n-c,a_n+c])$ is realizable.
To find the minimum $c$, we do a binary search with help of Theorem~\ref{theorem:CDZ} for verification; this takes $O(n \log n)$ time. 
Once optimal $c$ is known, the sequence $D$ can be computed
using Theorem~\ref{theorem:algo1_realizable}
to search graphic certificate in $\s_c$, thus the time complexity for computing 
sequence $D$ is $O(n\log n)$.

\begin{problem}
Given $A=(a_1,\ldots,a_n)$, find minimum fraction $\epsilon$ and a
graphic sequence $D=(d_1,\ldots,d_n)$ satisfying $a_i(1-\epsilon)\leq d_i \leq a_i(1+\epsilon)$.
\end{problem}
{\em Solution}:
Again we need to find smallest non-negative fraction
$\epsilon$ such that the interval sequence 
$\s_\epsilon =([a_1(1-\epsilon),a_1(1+\epsilon)],\ldots,[a_n(1-\epsilon),a_n(1+\epsilon)])$ is realizable.
To find the minimum $\epsilon$, we do a binary search with help of Theorem~\ref{theorem:CDZ}; 
this takes $O(n \log n)$ time. Once $\epsilon$ is known, 
using Theorem~\ref{theorem:algo1_realizable}, sequence $D$ can be computed
in $O(n\log n)$ time.


\newpage
\appendix
\begin{center}\LARGE{\bf Appendix~~}\end{center}

\section{An $O(n)$ time implementation of Theorem~\ref{theorem:EG} and Theorem~\ref{theorem:CDZ}}
\label{section:EG_CDZ_implementation}

Erd\"os and Gallai~\cite{EG60} presented a characterisation for graphic sequences; 
and later Cai et al.~\cite{CDZ:00} presented a similar characterisation for realizable interval-sequences
(Theorem~\ref{theorem:EG} and Theorem~\ref{theorem:CDZ}).
For completeness, we present simple $O(n)$ algorithm for efficiently computing vectors
$X(D)$, $Y(D)$, and $\varepsilon(\s )$ for any given sequence $D=(d_1,d_2,\ldots,d_n)$ 
and interval-sequence $\s =(A,B)=([a_1,b_1],\ldots,[a_n,b_n])$.
Recall 
$$X_k(D)= \sum_{i=1}^k d_i;~~
Y_k(D)=k(k-1)+\sum_{i=k+1}^n \min(d_i,k);~~
W_k(\s )=\{i\in[k+1,n] ~|~b_i\geq k+1\};$$
$$
	\varepsilon_k(\s ) =
	\begin{cases}
	1 & \text{if } a_i=b_i \text{ for } i\in W_k(\s ) \text{ and } \sum_{i\in W_k(\s )} (b_i+k|W_k(\s )|) \text{ is odd}.\\
    0 & \text{otherwise.}
	\end{cases}
$$

\noindent 
For any $k\in[1,n]$, $j(k)$ is defined to be $0$ if $k>d_1$, else $j(k)\in[1,n]$ is 
the largest index such that $d_{1}, d_{2}, \ldots, d_{j(k)}\geq k$.
Then $$Y_k(D)=k(k-1)+\sum_{i=k+1}^{j(k)} k + \sum_{i=j(k)+1}^{n} d_i=k(j(k)-1)+ \sum_{i=j(k)+1}^{n} d_i.$$
Observe that $j(k)$ is a non-increasing sequence.

\begin{algorithm}[!ht]
Let $D=(d_1,d_2,\ldots,d_n)$ be the input sequence\;
Set $X_1(D)=d_1$, $j(1)=n$, and $W_1(D)=0$\;
\For{$k=2$ to $n$}{
$X_k(D)=X_{k-1}(D)+d_k$\;
\While{$(d_{j(k)}\ngeq k)$}{
$j(k)=j(k)-1$\;
$W_k(D)=W_{k-1}(D)+d_{j(k)}$\;
}
$Y_k(D)=k(j(k)-1)+W_k(D)$\;
}
\Return $(X(D),Y(D))$.
\caption{Computation of vectors}
\label{Algorithm:vectors}
\end{algorithm}

It is easy to verify that the Algorithm~\ref{Algorithm:vectors} takes linear time, 
and correctly computes the vectors $X(D)$ and $Y(D)$.
Also similarly $\varepsilon(\s )$ is computable in $O(n)$ time.
So the following theorem is 
immediate.

\begin{theorem}
Given any sequence $D$ and interval-sequence $\s =(A,B)$ of length $n$, 
the vectors $X(D)$, $Y(D)$, and $\varepsilon(\s )$ are all computable in $O(n)$ time.
\label{theorem:implementation}
\end{theorem}


\end{document}